\documentclass[submission,copyright,creativecommons]{eptcs}
\providecommand{\event}{PLACES 2015} 
\usepackage{makeidx}

\usepackage{amsmath}
\usepackage{amsfonts}
\usepackage{amssymb}
\usepackage{mathpartir}
\usepackage{enumerate}
\usepackage{color}
\usepackage{stmaryrd} 
\usepackage{tikz}

\usepackage{times}
\usepackage{graphicx}

\title{A Typed Model for Dynamic Authorizations}

\author{
  Silvia Ghilezan 
  \qquad \qquad
  Svetlana Jak\v{s}i\'c 
   \qquad \qquad
  Jovanka Pantovi\'c 
      \institute{University of Novi Sad, Serbia}
  \and
  Jorge A. P\'erez 
  \institute{University of Groningen, The Netherlands}
  \and 
  Hugo Torres Vieira 
  \institute{IMT Institute for Advanced Studies Lucca, Italy}
}

\def\titlerunning{A Typed Model for Dynamic Authorizations} 

\def\authorrunning{Ghilezan, Jak\v{s}i\'c, Pantovi\'c, P\'{e}rez, Vieira}

\definecolor{ceca}{rgb}{1,0.5,0}

\newcommand{\vanja}[1]{#1}
\newcommand{\sg}[1]{#1}

\newcommand{\N}{{\cal N}}
\newcommand{\NA}{a}
\newcommand{\NB}{b}
\newcommand{\NC}{c}
\newcommand{\NX}{x}
\newcommand{\NY}{y}
\newcommand{\NZ}{z}
\newcommand{\fn}[1]{\mathsf{fn}(#1)}

\newcommand{\role}{\rho}

\newcommand{\PP}{P}
\newcommand{\PQ}{Q}
\newcommand{\PR}{R}
\newcommand{\inact}{0}
\newcommand{\parop}{\;|\;}
\newcommand{\rest}[1]{(\nu #1)}
\newcommand{\scope}[1]{(#1)}
\newcommand{\prefix}{\alpha}
\newcommand{\msg}{l}
\newcommand{\send}[4]{#1!#4}
\newcommand{\ssend}[2]{#1!#2}
\newcommand{\receive}[4]{#1?#4}
\newcommand{\rreceive}[2]{#1?#2}
\newcommand{\sauth}[4]{#1\langle#4\rangle}
\newcommand{\ssauth}[2]{#1\langle#2\rangle}
\newcommand{\rauth}[4]{#1(#4)}
\newcommand{\rrauth}[2]{#1(#2)}

\newcommand{\red}{\rightarrow}
\newcommand{\subst}[2]{\{#1/#2\}}

\newcommand\context{\mathcal{C}}

\definecolor{darkgreen}{rgb}{0.0, 0.6, 0}

\newtheorem{proposition}{Proposition}
\newtheorem{lemma}{Lemma}
\newtheorem{theorem}{Theorem}
\newtheorem{corollary}{Corollary}
\newtheorem{definition}{Definition}

\newenvironment{proof}[1][Proof]{\begin{trivlist}
\item[\hskip \labelsep {\bfseries #1}]}{\end{trivlist}}

\newcommand{\auth}[2]{\mathit{auth}(#1,#2)}

\newcommand{\rulename}[1]{\text{\small \textsc{#1}}}

\begin{document}
\maketitle

\begin{abstract}
Security requirements in distributed software systems are inherently dynamic.
In the case of authorization policies, resources are
meant to be accessed only by authorized parties, but the
authorization to access 
a resource may be dynamically granted/yielded. 
We 
describe ongoing work on a model for specifying communication
and dynamic authorization handling.
We build upon the $\pi$-calculus so as to enrich 
communication-based systems with authorization specification and
delegation; here authorizations regard channel usage and delegation 
refers to the act of yielding an authorization to another party.
Our 
model includes: (i) a novel scoping construct for
authorization, which allows to 
specify 
authorization boundaries,
and (ii) communication primitives for
authorizations, which allow to pass around authorizations to
act on a given channel.
An authorization error may consist in, e.g.,
performing an action along a name which is not under an appropriate
authorization scope.
We introduce a typing discipline that 
{ensures} that processes never 
reduce to authorization errors, even when authorizations are
dynamically delegated.
\end{abstract}

\section{Introduction}

Nowadays, computing systems operate in distributed environments, which may be highly heterogeneous,
including at the level of trustworthiness. It is often the case that collaborating systems need to protect
themselves from malicious entities by enforcing authorization policies that ensure actions are
carried out by properly authorized parties. Such \emph{authorizations} to act upon a resource may 
be statically prescribed --- for instance, a determined party is known to have a determined authorization ---
but may also be dynamically established --- for example, when a server delegates a task to a slave it 
may be sensible to pass along the appropriate authorization to carry out the delegated task. 

{As a motivating example, consider the message sequence chart given in Figure~\ref{fig:bankmsc} 
describing a scenario where a client interacts with a bank portal in order to request a credit.
After the client submits the request, the bank portal asks a teller to approve the request, allowing
him/her to join the ongoing interaction. Apart from some rating that could be automatically calculated
by the bank portal, it is the teller who ultimately approves/declines the request. 
It then seems reasonable that
the teller 
\emph{impersonates} the bank when informing the client 
about the outcome of the request.
At this point we may ask: is the teller authorized to 
act on behalf of the bank portal in this structured interaction? 
Even if the teller gained access to the communication medium when joining the interaction,
 the authorization to act on behalf of the bank portal may  not be necessarily granted; 
 in such cases an explicit mechanism that dynamically grants such an authorization is required.
To account for this kind of scenarios, 
in previous work~\cite{DBLP:journals/corr/GhilezanJPPV14} we explored the idea of \emph{role authorizations}. 
It appears to us that 
the key notions underlying this idea can be well explored in a more general setting; 
here we aim at distilling such notions in a simple setting.}

\begin{figure}[t]
 \begin{center} 
 \includegraphics[width=70mm]{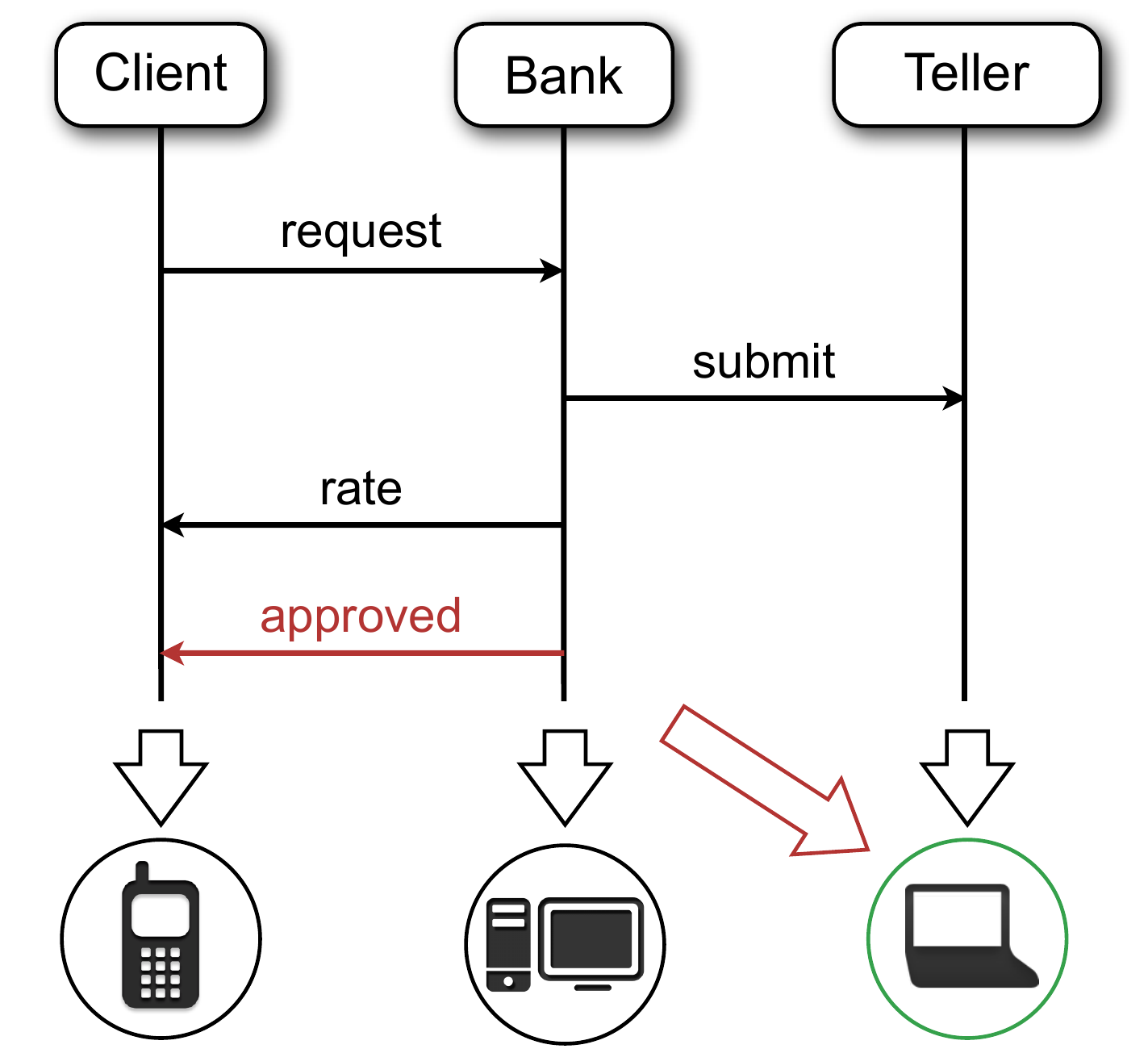}    
 \end{center}
\caption{Credit request scenario.}
\label{fig:bankmsc} 
\end{figure}

We distinguish an authorization from the resource itself: a system may already know the
identity of the resource (say, an email address or a file name) but may not be authorized 
to act upon it (e.g., is not able to send an email on behalf of a given address or to write on a file). Also, 
it might be the case that the system acquires knowledge about the resource (for instance, by receiving an 
email address or a file) but not necessarily is immediately granted access to act upon it. We 
focus on communication-centered systems in which authorizations are a first-class notion modeled in a 
dedicated way, minimally extending the $\pi$-calculus~\cite{DBLP:books/daglib/0004377} to 
capture dynamic authorization handling. As such, the resources that we consider are communication channels;
authorizations concern the ability to communicate on channels.

Authorizations may be associated with a \emph{spatial} connotation, as it seems fairly natural 
that 
a determined part of the system has access to 
a resource while the rest of the system does not. To this end, we introduce a \emph{scoping operator}
to specify delimited authorizations: we write $\scope{\NA} \PP$ to specify that process
$\PP$ is authorized to act upon the resource $\NA$. For example, by
$\scope{\NA} \send\NA\role\msg\NB.\PQ$ we specify a process that is authorized on channel $\NA$
and that is willing to use it to send $\NB$ after which it behaves as $\PQ$.

Also, since we are interested in addressing systems in which authorizations are dynamically passed 
around, we model authorization communication in a distinguished way by means of dedicated communication actions: we write $\sauth\NA\role\msg\NB.\PP$ to specify the action of sending an 
authorization to act upon $\NB$ (and proceeding as $\PP$) and $\rauth\NA\role\msg\NB.\PP$ to 
specify the action of receiving an authorization to act upon $\NB$ (and proceeding as $\PP$), where
in both cases channel $\NA$ is used as the underlying communication medium. We remark that both in
the authorization scoping $\scope{\NA} \PP$ and in the authorization reception 
$\rauth\NB\role\msg\NA.\PP$ the name $\NA$ is not bound so as to capture the notion that 
authorizations are handled at the level of known identities. 

Given the \vanja{sensitive} nature of an authorization, 
we believe it is natural to enforce a {\emph{specialized}} 
discipline regarding authorization manipulation. 
Namely, we consider that the act of passing along an authorization ---\emph{authorization delegation}---
entails the yielding of the communicated 
authorization.
That is, a party willing to communicate an authorization loses it after synchronization. 
Consider, for example, a process
$$S = \scope{\NA} ( \rauth\NA\role\msg\NB.\PP \parop \scope{\NB} \sauth\NA\role\msg\NB.\PQ)$$ 
while the process on the left-hand side of the parallel composition ($\parop$) is awaiting an authorization to act 
on $\NB$ (via a synchronization on channel $\NA$), the process on the right-hand side is willing to delegate 
authorization to act on $\NB$. In one reduction step, process $S$ evolves to 
$\scope{\NA} ( (\scope{\NB}\PP) \parop \PQ)$, 
thus capturing the fact that 
the process on the right ($\PQ$) has now
 lost the authorization to act on $\NB$. 
Notice that this authorization transfer may have influence on the
 resources already known to the receiving party (i.e., process $\PP$ may specify
communications on channel $\NB$). 

{The fact that authorizations are yielded when communicated allows us to model a form of 
authorization accounting, in the sense that authorizations are viewed as a ``countable'' resource.
As such, in general we would expect $\scope\NA\scope\NA\PP$ to differ from $\scope\NA\PP$. However, 
since we intuitively interpret $\scope\NA\PP$ as 
``the whole of $\PP$ is authorized to interact on $\NA$'',
it does not seem sensible that part of $\PP$ can completely yield the authorization. 
Consider, e.g., process $\scope\NA(\ssauth\NB\NA.\PP \parop  \PQ)$ where it does not seem reasonable
that the authorization delegation expressed by prefix $\ssauth\NB\NA$ interferes with the
authorization on $\NA$ already held by $\PQ$ which is (concurrently) active in the authorization
scope.
}
{Hence, a} \vanja{system cannot create/discard  valid authorizations (that are scoping active processes); authorizations can only float around.} 
It is also reasonable to 
{allow that} a party that delegates an authorization may get it back via another  synchronization step. 
This way, our model {allows for} reasoning about authorization ownership and lending. \vanja{Finally, we envisage (in our untyped model) a possibility for sharing a given authorization scope with a specified number of parties. This may be represented by specifying multiple copies of the same authorization scope. For example, process $\scope\NA\scope\NA\scope\NB\ssauth\NB\NA.\PP$ 
{(or $\scope\NA\scope\NB\ssauth\NB\NA.\scope\NA\PP$)} will retain authorization scope for $\NA$ and reduce to 
\sg{$\scope\NA\scope\NB\PP,$} after communication with $\scope\NB\rrauth\NB\NA.\PQ$}

Some previous works have explored dedicated scoping operators with security motivations~(see, e.g.,~\cite{Giunti,VivasYoshida}).
However, to our knowledge, 
the particular combination of a (non binding) scoping construct with 
name passing as in the $\pi$-calculus 
seems to be new. 
The syntactic elements of our process model, together with the dynamic nature of 
authorizations, pose challenges at the level of statically identifying processes 
that act only upon resources for which they are properly authorized. 
In this paper we start exploring 
a typing discipline for authorization manipulation that allows to statically {ensure} that processes 
 never incur in authorization errors, essentially by accounting process authorization requirements. 
In the remaining, we formally present the language and type system, and state our 
results. 

\section{Process Model}
\label{s:model}

\begin{table}[t]
\begin{math}
\displaystyle
\begin{array}[t]{@{}rcl@{\quad}l@{}}
  \PP,\PQ & ::= & \inact & \text{(Inaction)} \\
          & | & \PP\parop\PQ & \text{(Parallel)}\\
          & | & \rest\NA \PP & \text{(Restriction)}\\
          & | & \scope{\NA} \PP & \text{(Authorization)}\\
          & | & \prefix.\PP & \text{(Prefix)}
\end{array}
\qquad
\begin{array}[t]{@{}rcl@{\quad}l@{}}
  \prefix   & ::= & \send\NA\role\msg\NB & \text{(Output)} \\
            &  |  & \receive\NA\role\msg{x} & \text{(Input)}\\
            &  |  & \sauth\NA\role\msg\NB & \text{(Send authorization)}\\
            &  |  & \rauth\NA\role\msg\NB & \text{(Receive authorization)}\\
\end{array}
\end{math}
\vspace{-4mm}
\caption{\label{tab:syntax}Syntax of processes.}
\end{table}

We introduce our process 
calculus with authorization scoping and authorization delegation. Let $\N$ be a countable set of {\em names}, ranged over by $\NA,$ $\NB,$ $\NC, \ldots,$ $\NX,$ $\NY,$ $\NZ.$ The syntax of processes is given in Table~\ref{tab:syntax}. Processes $\inact,$ $\PP\parop\PQ,$ $\rest\NA\PP,$ $\send\NA\role\msg\NB.\PP$ and $\receive\NA\role\msg{x}.\PP$ comprise the usual $\pi$-calculus operators for specifying inaction, parallel composition, name restriction, and output and input communication actions, respectively. 
We introduce three novel operators, motivated earlier: 
\begin{enumerate}[1.]
\item $\sauth\NA\role\msg\NB.\PP$ sends an authorization for the name $\NB$ on $\NA$ and proceeds as $\PP$; 
\item $\rauth\NA\role\msg\NB.\PP$ receives an authorization for the name $\NB$ on $\NA$ and proceeds as $\PP$; 
\item $\scope{\NA} \PP$ authorizes all actions on the channel $\NA$ in $\PP$.
\end{enumerate} 
We remark on the novel reasoning regarding scope authorization $\scope{\NA}\PP$ in combination with
$\pi$-calculus-like name passing, since all actions on channel $\NA$ in process $\PP$ are authorized,
including actions originally specified for received names. For example, consider a process 
$\PP = \scope{\NB}\scope{\NA} \receive\NB\role\msg{x}.\send{x}\role\msg\NC.\inact$ that interacts in a context that sends
name $\NA$ on $\NB$.  Then $\PP$ may evolve to $\PP' =\scope{\NB} \scope{\NA} \send{\NA}\role\msg\NC.\inact$, which is authorization safe.
Still, authorizations may be ``revoked'' via authorization delegations. 

We introduce some auxiliary notions and abbreviations, useful for the remaining formal presentation.
%
\sg{The set 
of free names of a  process $\PP$, denoted 
$\fn\PP$, accounts for authorization constructs in the following way:}
\begin{eqnarray*}
\fn{\scope\NA \PP} & \triangleq &  \{a\} \cup \fn{\PP}  \\
 \fn{\sauth\NA\role\msg\NB.\PP} = \fn{\rauth\NA\role\msg\NB.\PP} & \triangleq &  \{a,b\} \cup\fn\PP
\end{eqnarray*}
Given a name $a$, 
we use $\alpha_\NA$ to refer to either
$\send\NA\role\msg\NB$, $\receive\NA\role\msg{x}$, $\sauth\NA\role\msg\NB$, or $\rauth\NA\role\msg\NB$. 
We abbreviate $(\nu a_1)(\nu a_2)\ldots(\nu a_k) P$ by  $(\nu \vec{a}) P$
and likewise $\scope{a_1}\scope{a_2} \ldots \scope{a_k} P$ by $\scope{\vec{a}} P$.

\begin{table}[t]
\[
\begin{array}{@{}c@{}} 
  \inferrule[]{}
  {\PP\parop\inact\equiv\PP}\qquad\quad
  \inferrule[]{}
  {\PP\parop\PQ\equiv\PQ\parop\PP}\qquad\quad
  \inferrule[]{}
  {(\PP\parop\PQ)\parop\PR\equiv\PP\parop(\PQ\parop\PR)}\qquad\quad
  \inferrule[]{}
  {\rest\NA\inact\equiv\inact}
   \vspace{2mm} \\
  \inferrule[]{}
  {\rest\NA\rest\NB\PP\equiv\rest\NB\rest\NA\PP}\qquad\quad
  \inferrule[]{}
  {\PP\parop\rest\NA\PQ\equiv\rest\NA(\PP\parop\PQ)\quad\text{if $\NA\notin\fn\PP$}}
    \qquad\quad
  \inferrule[]{}
  {\PP \equiv_{\alpha} \PQ \implies \PP \equiv \PQ}
 \vspace{2mm} \\
  \inferrule[]{}
  {\scope{\NA}\scope{\NB}\PP \equiv \scope{\NB}\scope{\NA}\PP}
    \qquad
  \inferrule[]{}
  {\scope{\NA}\inact \equiv \inact}
  \qquad
    \inferrule[]{}
  {\scope{\NA}(\PP \parop \PQ) \equiv \scope{\NA}\PP \parop \scope{\NA}\PQ}
\qquad
  \inferrule[]{}
  {\scope{\NA}\rest{\NB}\PP \equiv \rest{\NB}\scope{\NA}\PP \quad\text{if $\NA\neq \NB$}
  }
\end{array}
\]
\vspace{-4mm}
\caption{\label{tab:structural}Structural congruence.}
\end{table}

\emph{Structural congruence} 
expresses basic identities on the structure of processes; it is defined as
the least equivalence relation between processes that satisfies the rules given in
Table~\ref{tab:structural}. Apart from the usual identities for the static fragment of the $\pi$-calculus 
(cf. first seven rules in Table~\ref{tab:structural}), structural congruence 
gives basic principle for 
the novel authorization scope: (i) 
authorizations can be swapped around; (ii) authorizations can be discarded/created only for the inactive process; 
(iii) authorizations distributes over parallel composition; and (iv) authorizations and name restrictions 
can be swapped if the corresponding names differ. 
We remark that, differently from name restrictions, authorization 
scopes can \vanja{neither}  be extruded/confined: 
since authorizations are specified
for free names (that cannot be $\alpha$-converted), extruding/confining
authorizations
actually changes the meaning of processes. \vanja{For example, processes
$\scope\NB\rreceive\NB\NX.\ssend\NX\NB.\inact$ and $\scope\NA\scope\NB\rreceive\NB\NX.\ssend\NX\NB.\inact$ are not considered as structurally congruent, as the latter one authorizes the action on $\NA$ in case it receives $\NA$ through $\NB,$ while the former one does not.  Another distinctive property comes from the significance of multiplicity of authorization scopes. We do not adopt $\scope\NA\PP \equiv \scope\NA\scope\NA\PP$ for $\PP\neq\inact$}, {for the sake of authorization accounting}. 
%
%
Before presenting the operational semantics of the language, we ensure that the rewriting supported by structural
congruence is enough to isolate \vanja{top level} communication actions 
{together with their respective authorization scopes.} 

\begin{proposition}[Normal Form]
\label{pro:normalform}
For any process $\PQ$ we have that there are $\PP_1, \ldots,$ $\PP_k,$ $\prefix_1, \ldots,$ $\prefix_k,$ $\vec{c}$,
and $\vec{a}_1$, \ldots, $\vec{a}_k$, where $(\nu \vec{c})$ and $\scope{\vec{a}_i}$ for $i \in 1,\ldots,k$ can be empty sequences, 
such that
\begin{equation} \label{normal_form}
\PQ \equiv (\nu \vec{c}) (\scope{\vec{a}_1}\prefix_1.P_1 \parop \scope{\vec{a}_2}\prefix_2.P_2  \parop \ldots \parop \scope{\vec{a}_k}\prefix_k.P_k)
\end{equation}
\end{proposition}
\begin{proof}
(by induction on the structure of $\PQ$)
\\
$\PQ \equiv \inact:$ It is in the form (\ref{normal_form}).
\\
$\PQ \equiv \PQ' \parop \PQ'':$  By induction hypothesis, we have that
\[
\PQ'  \equiv  (\nu\vec{c})(\scope{\vec{a}_1}\prefix_1.P_1 \parop\ldots\parop\scope{\vec{a}_k}\prefix_k.P_k) \qquad
\PQ'' \equiv  (\nu\vec{d})(\scope{\vec{b}_1}\beta_1.Q_1    \parop \ldots \parop \scope{\vec{b}_l}\beta_l.Q_l)
\]
and we can assume, by application of $\alpha$-conversion, that $\vec{d} \cap \fn{\PQ'} =\emptyset.$ Therefore,
\[
\PQ \equiv (\nu\vec{c})(\nu\vec{d})
(
 \scope{\vec{a}_1}\prefix_1.P_1\parop\ldots\parop\scope{\vec{a}_k}\prefix_k.P_k \parop 
 \scope{\vec{b}_1}\beta_1.Q_1 \parop \ldots \parop \scope{\vec{b}_l} \beta_l.Q_l
).
\]
\item[] $\PQ\equiv \rest\NA \PP:$ Applying the induction hypothesis on $\PP$, we have that
\[
\PQ \equiv \rest\NA (\nu \vec{c}) (\scope{\vec{a_1}}\prefix_1.P_1 \parop \scope{\vec{a_2}}\prefix_2.P_2 
\parop \ldots \parop \scope{\vec{a_k}}\prefix_k.P_k).
\]
\item[] $\PQ \equiv \scope{\NA} \PP:$ By induction hypothesis and $\alpha$-conversion, we have that
\[
  \PQ \equiv \scope{\NA} (\nu \vec{c}) (\scope{\vec{a}_1}\prefix_1.P_1 \parop \scope{\vec{a}_2}\prefix_2.P_2 
\parop \ldots \parop \scope{\vec{a}_k}\prefix_k.P_k),
\] 
where $\NA\not \in \{c_1,\ldots,c_k\}.$ Hence,
\[
  \PQ \equiv (\nu \vec{c})(\scope{\NA}\scope{\vec{a}_1}\prefix_1.P_1 \parop \ldots \parop \scope{\NA}\scope{\vec{a}_k}\prefix_k.P_k).
\] 
$\PQ\equiv \prefix.\PP:$ It is in the form (\ref{normal_form}).
\end{proof}

\begin{table}[t]
\[
\begin{array}[t]{@{}c@{\qquad}c@{\qquad}c@{\qquad}c@{}}
\inferrule[(stru)]
{\PP \equiv \PP' \red \PQ' \equiv \PQ}
{\PP \red \PQ}
&
\inferrule[(parc)]
{\PP \red \PQ}
{\PP \parop \PR \red \PQ \parop \PR}
&
\inferrule[(newc)]
{\PP \red \PQ}
{\rest\NA \PP \red \rest\NA \PQ}
&
\inferrule[(autc)]
{\PP \red \PQ}
{\scope\NA \PP \red \scope\NA \PQ}
\\\\
\multicolumn{4}{c}
{\inferrule[(comm)]{}
{\scope{\vec{\NA}_1}\scope{\NB} \send\NB\role\msg{\NC}.\PP \parop \scope{\vec{\NA}_2}\scope{\NB}\receive\NB\role\msg{x}.\PQ
\red
\scope{\vec{\NA}_1}\scope{\NB} \PP \parop \scope{\vec{\NA}_2}\scope{\NB}\PQ\subst{\NC}{x}
}}
\\\\
\multicolumn{4}{c}
{\inferrule[(auth)]{}
{\scope{\vec{\NA}_1}\scope{\NB}\scope{\NC} \sauth\NB\role\msg{\NC}.\PP \parop \scope{\vec{\NA}_2}\scope{\NB}\rauth\NB\role\msg\NC.\PQ
\red
\scope{\vec{\NA}_1}\scope{\NB} \PP \parop \scope{\vec{\NA}_2}\scope{\NB}\scope{\NC}\PQ
}}
\end{array}
\]
\vspace{-4mm}
\caption{\label{tab:red}Reduction rules.}
\end{table}
We may then characterize the evolution of systems via a reduction relation, denoted by $\red$, defined as the least relation that satisfies  
the rules given in Table~\ref{tab:red}, focusing on the representative cases for synchronization and closing the relation under structural 
congruence $\rulename{(stru)}$ and static contexts $\rulename{(parc)}$, $\rulename{(newc)}$, and $\rulename{(autc)}$. 
Rule~$\rulename{(comm)}$ formalizes communication of names, stating that it can be performed only via authorized channel names 
--- notice we single out authorization scopes on channel $\NB$ both for output and input. Authorization delegation 
is formalized by rule~$\rulename{(auth)}$. It meets the following requirements: synchronization is realized via an authorized channel 
and the emitting process must have the authorization in order to delegate it away (names $\NB$ and $\NC$ in the rule, respectively);
after sending an authorization for a name the emitting process proceeds ($\PP$) falling outside of authorization scope of that name 
(losing authorization), and after receiving an authorization for a name the receiving process proceeds ($\PQ$) under the scope of the 
received authorization (acquiring authorization). Notice rules $\rulename{(comm)}$ and $\rulename{(auth)}$ address 
action prefixes up to the relevant authorizations (cf. Proposition~\ref{pro:normalform}). We denote by $\red^{\star}$ the reflexive and transitive closure of $\red.$

We introduce some auxiliary notions in order to syntactically characterize \emph{authorization errors} in our setting. 
First of all, we define the usual notion of \emph{active contexts} for our calculus:

\begin{definition}[Active Context]
\begin{math}
\displaystyle
\begin{array}[t]{@{}rcl@{\quad}l@{}}
  \context[\cdot] \; ::= \; \cdot \quad | \quad \PP\parop \context[\cdot] \quad
          | \quad \rest\NA \context[\cdot] \quad
          | \quad \scope{\NA} \context[\cdot]
\end{array}
\end{math}
\end{definition}

Active contexts allow us to talk about any active communication prefixes of a process.
Also, we may introduce a predicate that states that  an active context authorizes actions on a given channel.
More precisely, for a given context and a channel, when the hole of the context is filled with an action on the channel, it is authorised for that action.

\begin{definition}[Context Authorization]
For an active context $\context[\cdot]$ and a channel $\NA,$ the context authorization predicate, denoted $\auth{\context[\cdot]}{\NA}$, is defined inductively on the structure of $\context[\cdot]$ as
\begin{equation*}
\auth{\context[\cdot]}{\NA} 
\triangleq
\ \left\{\begin{array}{@{}l@{\quad}l@{}}
\mathit{false} & \text{if } \context[\cdot] = \cdot\\
\mathit{true} & \text{if } \context[\cdot] = \scope\NA\context'[\cdot] \\
\auth{\context'[\cdot]}{\NA} & \text{if } \context[\cdot] = \scope\NB\context'[\cdot] \text{ and } \NA \neq \NB\\
\auth{\context'[\cdot]}{\NA} & \text{if } \context[\cdot] = \PP \parop \context'[\cdot] \\
\auth{\context'[\cdot]}{\NA} & \text{if } \context[\cdot] = \rest\NB\context'[\cdot] 
\end{array}\right.
\end{equation*}
\end{definition}

We may then use active contexts and context authorization to precisely characterize errors in our model,
since active contexts allow us to talk about any communication prefix in the process and the context 
authorization predicate allows to account for the authorizations granted by the context:
processes that have active communication prefixes which are not authorized are errors.

\begin{definition}[Error]\label{d:error}
We say process $\PP$ is an error if $P \equiv \context[\alpha_\NA.\PQ]$ where 
\begin{enumerate}
\item $ \auth{\context[\cdot]}{\NA}=\mathit{false}$ or 
\item $\alpha_{\NA} = \sauth\NA\role\msg\NB$ and
$\auth{\context[\cdot]}{\NB}=\mathit{false}$.
\end{enumerate}
\end{definition}

Notice that by $\alpha_\NA$ we refer to any communication action on channel $\NA$, so {intuitively
communication actions that may cause stuck configurations according to our semantics (where
synchronizations only occur when processes hold the proper authorizations) are seen as errors.}

\section{Type System}\label{s:types}
In order to statically single out the processes that can never evolve into authorization errors, 
we introduce a suitable type system that accounts for the authorizations required by the processes.

\paragraph{Typing Judgment and Typing Rules.}
Let  $\rho$ denote a set of names.
The typing judgment $\rho \vdash \PP$
states that process $\PP$ is typed if the context provides authorizations for names $\rho$; hence the process 
performs actions in the (unauthorized) names $\rho$ (i.e., actions along names not under respective 
authorization scopes). 
Thus, $\emptyset \vdash \PP$ says that all communication actions prescribed by 
process $\PP$ are authorized, i.e., occur within the scope of the appropriate authorizations.
We say that process $\PP$ is {\em well typed} if $\emptyset \vdash \PP$.

Typing rules are given in Table~\ref{tab:typing}.
The inactive process contains no actions along unauthorized channel names $(\rulename{tstop}).$
If two processes act along unauthorized channel names $\rho_1$ and $\rho_2,$ their parallel composition 
performs actions along the union $\rho_1\cup\rho_2$ of unauthorized names  $(\rulename{tpar}).$
If a typed process $\PP$ does not perform actions along a channel $\NA,$ the process where name $\NA$ 
is restricted is typed under the same set of names as $\PP$ $(\rulename{tnew}).$
If $\PP$ acts under a set of unauthorized names $\rho,$ then $\scope\NA\PP$ authorizes  
$\NA$ in $\PP$ and thus performs actions under the set of unauthorized names $\rho\setminus\{\NA\}$ $(\rulename{tauth}).$
Sending a name along a channel $\NA$ extends the set of unauthorized names with the name $\NA$ $(\rulename{tsend}).$
Receiving a name $\NX$ along a channel  $\NA$  extends the set of unauthorized names with the name $\NA$ and it is typed only if there is no unauthorized actions along $\NX$ within $\PP$ $(\rulename{trecv}).$
Sending authorization for a name $\NB$ along a channel $\NA$ extends the set of unauthorized names with both names $\NA$ and $\NB$ $(\rulename{tdeleg}).$
Receiving authorization for a name $\NB$ along a name $\NA$ is typed under the set of unauthorized names that is extended with $\NA$ and does not contain $\NB$ (since the reception authorizes $\NB$ in $\PP$).  

\begin{table}[t]
\[
\begin{array}[t]{@{}c@{\qquad}c@{\qquad}c@{\qquad}c@{}}
\inferrule[(tstop)]
{}{\emptyset \vdash \inact}
&
\inferrule[(tpar)]
{\rho_1 \vdash \PP \quad \rho_2 \vdash \PQ}
{\rho_1 \cup \rho_2 \vdash \PP \parop \PQ}
&
\inferrule[(tnew)]
{\rho \vdash \PP \quad \NA \not \in \rho}
{\rho \vdash \rest\NA \PP}
&
\inferrule[(tauth)]
{\rho \vdash \PP}
{\rho \setminus \{\NA\} \vdash \scope\NA\PP}
\\\\
\inferrule[(tsend)]
{\rho \vdash \PP}
{\rho \cup \{\NA\} \vdash \send\NA\role\msg\NB. \PP}
&
\inferrule[(trecv)]
{\rho \vdash \PP \quad x \not \in \rho}
{\rho \cup \{\NA\} \vdash \receive\NA\role\msg{x}. \PP}
& 
\inferrule[(tdeleg)]
{\rho \vdash \PP \quad \NB \not \in \rho}
{\rho \cup \{\NA,\NB\} \vdash \sauth\NA\role\msg\NB. \PP}
&
\inferrule[(trecp)]
{\rho \vdash \PP }
{(\rho \setminus \{\NB\})\cup \{\NA\} \vdash \rauth\NA\role\msg\NB. \PP}
\end{array}
\]
\vspace{-4mm}
\caption{\label{tab:typing}Typing rules.}
\end{table}

\paragraph{Main Results.}
Our main result is \emph{type safety}: well-typed processes never evolve into an (authorization) error, 
in the sense of Definition~\ref{d:error}. This is  stated as Corollary~\ref{cor:safety}; 
before giving the main statement we show its supporting results.
We first state a basic property of typing derivations: unauthorized names must be included
in the free names of the process. 
\begin{proposition} 
\label{pro:freenames}
If $\rho \vdash \PP$ then $\rho \subseteq \fn{\PP}$.
\end{proposition}
\begin{proof}
(by induction on the depth of the derivation of $\rho \vdash \PP$)
\\
If $\emptyset\vdash\inact$ then $\fn{\inact}=\emptyset.$
\\
The following cases follow by definition of free names and the induction hypothesis.
\\
Case $\rho_1\cup\rho_2\vdash \PP\parop\PQ$ is derived from $\rho_1\vdash\PP$ and $\rho_2\vdash\PQ:$  $\fn{\PP\parop\PQ}=\fn{\PP}\cup\fn{\PQ} \supseteq \rho_1\cup\rho_2.$
\\
Case $\rho\vdash\rest\NA\PP$ is derived from $\rho\vdash\PP$ and $\NA\not\in\rho:$  $\fn{\rest\NA\PP}=\fn{\PP}\setminus\{\NA\}\supseteq \rho\setminus\{\NA\}=\rho.$
\\
Case  $\rho\setminus\{\NA\}\vdash\scope\NA\PP$ is derived from $\rho\vdash\PP:$ $\fn{\scope\NA\PP}=\fn{\PP}\cup\{\NA\}\supseteq \rho\cup\{\NA\}\supseteq \rho\setminus\{\NA\}.$
\\
Case $\rho\cup\{\NA\}\vdash\send\NA\role\msg\NB. \PP$ is derived from $\rho\vdash\PP:$ $\fn{\send\NA\role\msg\NB. \PP}=\fn{\PP}\cup\{\NA\}\supseteq \rho\cup\{\NA\}.$
\\
Case $\rho \cup \{\NA\} \vdash \receive\NA\role\msg{x}. \PP$ is derived from $\rho\vdash\PP$ and $\NX\not\in \rho:$
$\fn{\receive\NA\role\msg{x}. \PP}=(\fn{\PP}\setminus\{\NX\})\cup\{\NA\} \supseteq (\rho \setminus\{\NX\})\cup\{\NA\} =\rho\cup\{\NA\}.$
\\
Case $\rho \cup \{\NA,\NB\} \vdash \sauth\NA\role\msg\NB. \PP$ is derived from $\rho\vdash\PP$ and $\NB\not\in \rho:$
$\fn{\sauth\NA\role\msg\NB. \PP}=\fn{\PP}\cup\{\NA,\NB\} \supseteq \rho\cup\{\NA,\NB\}.$
\\
Case $(\rho \setminus \{\NB\})\cup \{\NA\} \vdash \rauth\NA\role\msg\NB. \PP$ is derived from $\rho\vdash\PP:$
$\fn{\rauth\NA\role\msg\NB. \PP}=\fn{\PP}\cup\{\NA,\NB\}\supseteq \rho\cup\{\NA,\NB\} \supseteq (\rho \setminus \{\NB\})\cup \{\NA\}.$
\end{proof}

We now state results used to prove that typing is preserved under system evolution, namely that
(i) typing is preserved under structural congruence, as reduction is closed under structural congruence,
and that (ii) typing is preserved under name substitution, since channel passing involves name substitution.

\begin{lemma}[Inversion Lemma] \label{lem:inversion}
\begin{enumerate}[1.]
\item If $\rho \vdash \inact$ then $\rho=\emptyset.$
\item If $\rho \vdash \PP \parop \PQ$ then there are  $\rho_1$ and $\rho_2$ such that $\rho=\rho_1\cup\rho_2$ and $\rho_1\vdash\PP$ and $\rho_2\vdash\PQ.$
\item If $\rho \vdash \rest\NA \PP$ then $\rho \vdash \PP$ and $\NA \not \in \rho.$
\item If $\rho \vdash \scope\NA\PP$ then there is $\rho'$ such that $\rho=\rho'\setminus\{a\}$ and $\rho' \vdash \PP.$
\item If $\rho  \vdash \send\NA\role\msg\NB. \PP$ then there is $\rho'$ such that $\rho=\rho'\cup\{a\}$ and $\rho' \vdash \PP.$
\item If $\rho \vdash \receive\NA\role\msg{x}. \PP$ then there is $\rho'$ such that $\rho=\rho'\cup\{\NA\}$ and $\NX\not\in\rho'$ and $\rho' \vdash \PP.$
\item If $\rho \vdash \sauth\NA\role\msg\NB. \PP$ then there is $\rho'$ such that $\rho=\rho'\cup\{\NA,\NB\}$ and $\NB\not\in \rho'$ and $\rho' \vdash \PP.$
\item If $\rho  \vdash \rauth\NA\role\msg\NB. \PP$ then there is $\rho'$ such that $\rho=(\rho'\setminus \{\NB\})\cup\{\NA\}$ and $\rho' \vdash \PP.$
\end{enumerate}
\end{lemma}

\begin{lemma}[Subject Congruence] 
\label{lem:subcong}
If $\rho \vdash \PP$ and $\PP \equiv \PQ$ then $\rho \vdash \PQ$.
\end{lemma}
\begin{proof}
(by induction on the depth of the derivation of $\PP \equiv \PQ$)
\\\\
We only write the following two interesting cases, and other cases can be obtained by similar reasoning.
\\\\
Case $\PP\parop\rest\NA\PQ\equiv\rest\NA(\PP\parop\PQ)$ and $\NA\notin\fn\PP:$
\\\\
From $\rho \vdash \PP\parop\rest\NA\PQ,$ by Lemma~\ref{lem:inversion}.~2, there are $\rho_1$ and $\rho_2$ such that 
$\rho_1 \vdash \PP$ and $\rho_2 \vdash \rest\NA\PQ.$
Therefore, by Lemma~\ref{lem:inversion}.~3, $\rho_2\vdash\PQ$ and $\NA\not\in\rho_2.$
Since $\NA\notin \fn\PP$ we conclude by Proposition~\ref{pro:freenames} that $\NA\notin\rho_1.$
By the rule $\rulename{(tpar)}$ we get that $\rho_1\cup\rho_2\vdash \PP\parop\PQ,$ and from $\NA\not\in \rho_1\cup\rho_2,$  by $\rulename{(tnew)}$ we derive
$\rho\vdash\rest\NA(\PP\parop\PQ).$
\\\\
Case $\scope{\NA}\rest{\NB}\PP \equiv \rest{\NB}\scope{\NA}\PP$ and  $\NA\neq \NB:$
\\\\
If $\rho\vdash\scope{\NA}\rest{\NB}\PP$ then by Lemma~\ref{lem:inversion}.~4 there is $\rho'$ such that $\rho'\vdash \rest{\NB}\PP$ and $\rho=\rho'\setminus\{\NA\}.$
By Lemma~\ref{lem:inversion}.~3, $\rho'\vdash\PP$ and $\NB\not\in\rho'$ (and so $\NB\not\in\rho'\setminus\{\NA\}$). Hence, by $\rulename{(tauth)}$ and $\rulename{(tnew)},$ we get $\rho \vdash \rest{\NB}\scope{\NA}\PP.$
\end{proof}

\begin{lemma}[Substitution]
\label{lem:subst}
If $\rho \vdash \PP$ then $\rho \subst{a}{b}Ê\vdash \PP \subst{a}{b}$.
\end{lemma}
\begin{proof}
(by induction on the depth of the derivation of $\rho \vdash \PP$)
\\\\
We give only one interesting case.
If $\rho\cup\{\NB,\NC\}\vdash\sauth\NC\role\msg\NB. \PP$ is derived from $\rho\vdash\PP$ and $\NB\not\in\rho.$
It holds that $(\rho\cup\{\NB,\NC\})\subst{\NA}{\NB}=\rho\cup\{\NA,\NC\}$ and 
$(\sauth\NC\role\msg\NB.\PP)\subst{\NA}{\NB}=\sauth\NC\role\msg\NA.\PP\subst{\NA}{\NB}.$
By induction hypothesis, $\rho\subst{\NA}{\NB}\vdash\PP\subst{\NA}{\NB},$ and by the rule $(\rulename{tdeleg}),$ $\rho \cup\{\NA,\NC\}\vdash \sauth\NC\role\msg\NA.\PP\subst{\NA}{\NB}.$
\end{proof}

We may now state our soundness result which ensures typing is preserved under reduction.
\begin{theorem}[Subject reduction] 
\label{the:subred}
If $\rho \vdash \PP$ and $\PP \red \PQ$ then $\rho \vdash \PQ$.
\end{theorem}
\begin{proof}
(by induction on the depth of the derivation of $\PP \red \PQ$)
\\\\
\underline{Base case 1:}
Assume that $\rho\vdash  \scope{\vec{\NA}_1}\scope{\NB} \send\NB\role\msg{\NC}.\PP \parop \scope{\vec{\NA}_2}\scope{\NB}\receive\NB\role\msg{x}.\PQ$ and 
\[
\rulename{(comm)}
\qquad
\scope{\vec{\NA}_1}\scope{\NB} \send\NB\role\msg{\NC}.\PP \parop \scope{\vec{\NA}_2}\scope{\NB}\receive\NB\role\msg{x}.\PQ
\red
\scope{\vec{\NA}_1}\scope{\NB} \PP \parop \scope{\vec{\NA}_2}\scope{\NB}\PQ\subst{\NC}{x}.
\]
By Lemma~\ref{lem:inversion}.~2, there are $\rho_1$ and $\rho_2$ such that $\rho=\rho_1\cup\rho_2$ and 
 \[
 \rho_1 \vdash  \scope{\vec{\NA}_1}\scope{\NB} \send\NB\role\msg{\NC}.\PP \quad \mbox{ and } \quad  \rho_2\vdash\scope{\vec{\NA}_2}\scope{\NB}\receive\NB\role\msg{x}.\PQ.
 \]
By consecutive application of Lemma~\ref{lem:inversion}.~4, there are $\rho_1'$ and $\rho_2'$ such that $\rho_1=\rho_1'\setminus \{\vec{a_1}, \NB\}$ and $\rho_2=\rho_2'\setminus \{\vec{a_2}, \NB\}$ and 
\[
  \rho_1' \vdash \send\NB\role\msg{\NC}.\PP \quad \mbox{ and } \quad  \rho_2' \vdash \receive\NB\role\msg{x}.\PQ.
\]
By Lemma~\ref{lem:inversion}.~5-6, there are $\rho_1''$ and $\rho_2''$ such that $\rho_1'=\rho_1''\cup\{\NB\}$ and $\rho_2'=\rho_2''\cup\{\NB\}$ and $\NX\not\in\rho_2''$ and
\[
  \rho_1'' \vdash \PP \quad \mbox{ and } \quad  \rho_2'' \vdash \PQ.
\]
One should notice that $\rho_1=(\rho_1''\cup\{\NB\})\setminus\{\vec{a_1},\NB\}=\rho_1''\setminus\{\vec{a_1},\NB\}$ and 
$\rho_2=(\rho_2''\cup\{\NB\})\setminus\{\vec{a_2},\NB\}=\rho_2''\setminus\{\vec{a_2},\NB\}$ and $\rho_2''\subst{c}{x}=\rho_2''$ (since $\NX\not\in\rho_2'').$
By Lemma~\ref{lem:subst} and consecutive application of the typing rules $\rulename{(tauth)}$ we get
\[
 \rho_1 \vdash  \scope{\vec{\NA}_1}\scope{\NB} \PP \quad \mbox{ and } \quad  \rho_2\vdash\scope{\vec{\NA}_2}\scope{\NB}.\PQ\subst{\NC}{\NX}.
\]
and finally, by $\rulename{(tpar)},$ we have  $\rho \vdash  \scope{\vec{\NA}_1}\scope{\NB} \PP \parop \scope{\vec{\NA}_2}\scope{\NB}.\PQ\subst{\NC}{\NX}.$
\\\\
\underline{Base case 2:} Assume that $\rho\vdash  \scope{\vec{\NA}_1}\scope{\NB}\scope{\NC} \sauth\NB\role\msg{\NC}.\PP \parop \scope{\vec{\NA}_2}\scope{\NB}\rauth\NB\role\msg\NC.\PQ$ and 
\[
\rulename{(auth)}
\qquad
\scope{\vec{\NA}_1}\scope{\NB}\scope{\NC} \sauth\NB\role\msg{\NC}.\PP \parop \scope{\vec{\NA}_2}\scope{\NB}\rauth\NB\role\msg\NC.\PQ
\red
\scope{\vec{\NA}_1}\scope{\NB} \PP \parop \scope{\vec{\NA}_2}\scope{\NB}\scope{\NC}\PQ
\]
By Lemma~\ref{lem:inversion}.~2, there are $\rho_1$ and $\rho_2$ such that $\rho=\rho_1\cup\rho_2$ and 
\[
  \rho_1 \vdash \scope{\vec{\NA}_1}\scope{\NB}\scope{\NC} \sauth\NB\role\msg{\NC}.\PP 
  \quad\mbox{ and }\quad
  \rho_2 \vdash \scope{\vec{\NA}_2}\scope{\NB}\rauth\NB\role\msg\NC.\PQ
\]
By consecutive application of Lemma~\ref{lem:inversion}.~4, there are $\rho_1'$ and $\rho_2'$ such that $\rho_1=\rho_1'\setminus \{\vec{a_1}, \NB, \NC \}$ and $\rho_2=\rho_2'\setminus \{\vec{a_2}, \NB\}$ and 
\[
  \rho_1' \vdash \sauth\NB\role\msg{\NC}.\PP 
  \quad\mbox{ and }\quad
  \rho_2' \vdash \rauth\NB\role\msg\NC.\PQ
\]
By Lemma~\ref{lem:inversion}.~7-8, there are $\rho_1''$ and $\rho_2''$ such that $\rho_1'=\rho_1''\cup\{\NB, \NC\}$ and $\NC\not\in\rho_1''$ and 
$\rho_2'=(\rho_2''\setminus \{\NC\}) \cup\{\NB\}$  and
\[
  \rho_1'' \vdash \PP 
  \quad\mbox{ and }\quad
  \rho_2'' \vdash \PQ.
\]
We conclude that $\rho_1=(\rho_1''\cup\{\NB,\NC\}) \setminus\{\vec{a_1},\NB,\NC\}=\rho_1''\setminus \{\vec{a_1},\NB\}$ (since $\NC\not\in \rho_1''$) and $\rho_2=\rho_2''\setminus \{\vec{a_2},\NB,\NC\}.$ By consecutive application of the typing rule $\rulename{(tauth)},$ we get
\[
  \rho_1 \vdash \scope{\vec{\NA}_1}\scope{\NB} \PP  
  \quad \mbox{ and }\quad  
  \rho_2 \vdash \scope{\vec{\NA}_2}\scope{\NB}\scope{\NC}\PQ
\]
and by $\rulename{(tpar)}$
\[
  \rho \vdash \scope{\vec{\NA}_1}\scope{\NB} \PP \parop \scope{\vec{\NA}_2}\scope{\NB}\scope{\NC}\PQ.
\]
\\\\
\underline{Case 3:}
Assume that $\rho \vdash \PP \parop \PR$ and $\PP \parop \PR \red \PQ \parop \PR$ is derived from $\PP\red\PQ.$ By Lemma~\ref{lem:inversion}.~2, there are $\rho_1$ and $\rho_2$ such that $\rho=\rho_1\cup\rho_2$ and
\[
  \rho_1 \vdash \PP \quad \mbox{ and } \quad \rho_2 \vdash \PR.
\]
By induction hypothesis, it holds that $\rho_1\vdash \PQ$ and therefore we have, by $\rulename{(tpar)},$ that $\rho \vdash \PQ\parop\PR.$
\\\\
\underline{Case 4:}
Assume that $\rho \vdash \rest\NA \PP$ and $\rest\NA \PP \red \rest\NA \PQ$ is derived from $\PP\red\PQ.$ By Lemma~\ref{lem:inversion}.~3, it holds that $\NA\not\in\rho$ and $\rho \vdash \PP.$ 
By induction hypothesis, it holds that $\rho\vdash \PQ$ and therefore, by $\rulename{(tnew)},$ we get $\rho \vdash \rest\NA\PQ.$
\\\\
\underline{Case 5:}
Assume that $\rho \vdash \scope\NA \PP$ and $\scope\NA \PP \red \scope\NA \PQ$ is derived from $\PP\red\PQ.$ By Lemma~\ref{lem:inversion}.~4, it holds that there is $\rho'$ such that $\rho=\rho'\setminus\{\NA\}$ and $\rho' \vdash \PP.$ By induction hypothesis, it holds that $\rho'\vdash \PQ$ and therefore, by $\rulename{(tauth)},$ we get $\rho \vdash \scope\NA\PQ.$
\\\\
\underline{Case 6:}
Assume that $\rho\vdash\PP$ and $\PP \red \PQ$ is derived from $\PP' \red \PQ',$  where $\PP \equiv \PP'$ and $ \PQ' \equiv \PQ.$
We conclude by Lemma~\ref{lem:subcong} that $\rho\vdash \PP'.$ Than, by induction hypothesis, we get $\rho\vdash \PQ',$ and applying again Lemma~\ref{lem:subcong}, we have that $\rho\vdash\PQ.$
\end{proof}

Theorem~\ref{the:subred} ensures, considering $\rho= \emptyset$, that well-typed processes
always reduce to well-typed processes. We now express the basic property for well-typed systems, namely
that they do not expose any authorization errors up to the ones granted by pending authorizations $\rho$.
\begin{proposition}[Error Free]
\label{pro:errorfree}
If $\rho \vdash \context[\alpha_\NA.\PQ]$ and $\NA \not \in \rho$ 
then $\auth{\context[\cdot]}{\NA}$.
\end{proposition}
\begin{proof}
(by induction on the structure of $\context[\cdot]$)
\\\\
Case $\context[\cdot]=[\cdot]:$
If  $\rho \vdash \alpha_\NA.\PQ$ we conclude that $\NA\in\rho,$ by Lemma~\ref{lem:inversion}.
\\
Case $\context[\cdot]=\PP\parop \context'[\cdot]:$
If $\rho\vdash\PP\parop \context'[\alpha_\NA.\PQ],$ by Lemma~\ref{lem:inversion}.~2, there are $\rho_1$ and $\rho_2$ such that $\rho=\rho_1\cup\rho_2$ and $\rho_1\vdash \PP$ and $\rho_2 \vdash \context'[\alpha_\NA.\PQ].$ By induction hypothesis, $\auth{\context'[\cdot]}{\NA},$ while by definition
$\auth{\PP\parop\context'[\cdot]}{\NA}=\auth{\context'[\cdot]}{\NA}.$
\\
Case $\context[\cdot]=\rest\NB\context'[\cdot]:$ 
If $\NA\not\in\rho$ and $\rho\vdash \rest\NB\context'[\alpha_\NA.\PQ],$ by Lemma~\ref{lem:inversion}.~3,  $\rho\vdash\context'[\alpha_\NA.\PQ]$ and $\NB\not\in\rho.$ By induction hypothesis, $\auth{\context'[\cdot]}{\NA}.$ By definition, $\auth{\rest\NB\context'[\cdot]}{\NA}=\auth{\context'[\cdot]}{\NA}.$
\\
Case $\context[\cdot]=\scope\NB\context'[\cdot]$ and $\NA\neq\NB:$
If $\NA\not\in\rho$ and $\rho\vdash \scope\NB\context'[\alpha_\NA.\PQ],$ by Lemma~\ref{lem:inversion}.~4, there is $\rho'$ such that $\rho=\rho'\setminus\{\NB\}$ and $\rho'\vdash\context'[\alpha_\NA.\PQ].$ If $\NA\neq\NB$ and $\NA\not\in\rho$ then $\NA\not\in\rho'.$ By induction hypothesis, $\auth{\context'[\cdot]}{\NA}.$ By definition, $\auth{\scope\NB\context'[\cdot]}{\NA}=\auth{\context'[\cdot]}{\NA}.$
\\
Case $\context[\cdot]=\scope\NA\context'[\cdot]:$
By definition $\auth{\scope\NA\context'[\cdot]}{\NA}=\mathit{true}.$
\end{proof}

Proposition~\ref{pro:errorfree} thus ensures that active communication prefixes that do not involve a pending 
authorization (outside of $\rho$) are not errors. Considering $\rho = \emptyset$ we thus have that well-typed 
processes do not have any unauthorized prefixes and thus are not errors. Along with Theorem~\ref{the:subred} 
we may then state our safety result which says well-typed processes never evolve into an error.

\begin{corollary}[Type Safety]
\label{cor:safety}
If $\;\emptyset \vdash \PP$ and $\PP \red^{\star} \PQ$ then $\PQ$ is not an error.
\end{corollary}
\begin{proof} Immediate from Theorem~\ref{the:subred} and Proposition~\ref{pro:errorfree}.
\end{proof}

Corollary~\ref{cor:safety}  attests that well-typed systems never reduce to authorization errors, including 
when authorizations are dynamically delegated. The presented type system allows for a streamlined analysis on
process authorization requirements, which we intend to exploit as the building block for 
richer analysis. 

\section{Concluding Remarks} 
The work presented here builds on our previous work~\cite{DBLP:journals/corr/GhilezanJPPV14}, in which we explored authorization passing
in the context of communication-centered systems. In~\cite{DBLP:journals/corr/GhilezanJPPV14}, the analysis addressed not only authorization passing
but also role-based protocol specification, building on the conversation type analysis presented 
in~\cite{DBLP:conf/tgc/BaltazarCVV12}. Here we focussed exclusively on the authorization problem,
obtaining a simple model which paves the way for further investigation, since the challenges
involved may now be highlighted in a crisper way.
{As usual,} 
\vanja{there are non typable processes that are authorized for all the actions and reduce to $\inact$. 
This is unsurprising, given the simplicity of the analysis.
An example is 
\[
  \scope\NA\scope\NB(\rreceive\NB\NX.\ssend\NX\NB.\inact\parop\ssend\NB\NA.\rreceive\NA\NX.\inact).
\]
For the same reason, even \sg{though} our untyped model enables to keep existing copies of delegated authorization scopes,  the type system restricts the usage of their scopes.
For example, the current discipline can not type the process
\[
  Ê\scope\NB\scope\NA\scope\NA\ssauth\NB\NA\ssend\NA\NB.\inact \parop \scope\NB\rrauth\NB\NA\rreceive\NA\NX.\inact
\]
even \sg{though} it safely reduces to $\inact.$
}
We believe it would be interesting to
enrich the typing analysis so that it encompasses the contextual information (authorizations already
held by the process) so as to address name reception \vanja{and authorization delegation} in a different way.
We also believe it would be interesting to discipline 
authorization usage so as to ensure 
absence of double authorizations {for the sake of authorization accountability, so as to ensure only 
the strictly necessary authorizations are specified.}

Naturally, it would also be interesting
to integrate the analysis presented here in richer settings, for instance (i)~considering the need to ensure
protocol fidelity using session types~\cite{DBLP:conf/esop/HondaVK98}, or (ii)~ensuring 
liveness properties so that security critical events are guaranteed to take place, or (iii)~exploring an 
ontology on names so that authorizations to act upon higher ranked names automatically yield 
authorization for lower-level ones. 
While relevant, these extensions appear as orthogonal developments to 
the analysis presented here and therefore should be studied in depth in a dedicated way.

We briefly review some related works. Scoping operators have been widely used for the purpose of modeling security aspects (e.g.,~\cite{Giunti}) but typically they use bound names (e.g., to model  secrets).
With the aim of representing secrecy and confidentiality requirements in process specifications,
an alternative scoping operator called \emph{hide} is investigated in \cite{Giunti}.
The hide operator is embedded in the so-called secrecy $\pi$-calculus, tailored to program secrecy in 
communications. The expressiveness of the hide operator is investigated in the context of a
behavioral theory, by means of an Spy agent. 
In contrast, our (free name) scoping operator focuses on authorization, a different security concern. 
In~\cite{VivasYoshida}, 
a scoping operator (called \emph{filter}) is proposed 
for dynamic channel screening. In a different setting 
(higher-order communication) and with similar properties, the filter
operator blocks all the actions that are not contained in the corresponding
filter (which contains polarised channel names). Contrary to the
authorization scope, filters are statically assigned to processes, while the
authorization scope assigned to a process may be dynamically changed.
To the best of our knowledge, the authorization scoping proposed here 
has not been explored before for the specification of communication-centered systems.

{One key idea explored here is the separation between resource and the respective 
authorization which in particular allows us to distinguish the communication of the resource handler from
the resource authorization. Consider, e.g., process
$
 \rreceive\NA\NX. \ssend\NB\NX. \inact
$,
where a forwarding process receives a channel and forwards it without necessarily
being authorized to interact in it. This allows to model scenarios where resource handlers may
be passed around in unverified contexts, since their unauthorized use is excluded. This example
in particular distinguishes the type-based authorization handling presented 
in~\cite{DBLP:journals/jlp/GorlaP09}, since authorizations directly flow in the 
communications (via the types) and cannot be received afterwards --- we leave to future work
a comparison with more \emph{refined} typing notions such as~\cite{DBLP:conf/sefm/FrancoV13,DBLP:conf/esop/SwamyCC10}
where we conceive that dependencies between received values can address such a separation,
albeit in a more indirect way.

To further remark on the particularities of our linguistic constructors consider process
\[
 \rreceive\NA\NX. \scope\NX \PP
\]
where $\PP$ is authorized to act on channel $\NX$, regardless of the identity of the channel
actually received, which somewhat hints on the particular combination between the
(non-binding) scoping operator and name passing. 
While we do not claim that our constructs cannot be encoded in other models,  we do
believe they provide an adequate abstraction level to reason on authorization handling. Still, it would be 
interesting to see how to express authorization scopes and authorization communication
(including delegation) using $\pi$-calculus like models (such as, e.g.,~\cite{appliedpi,polipi,pigroups}).
}

There are high-level similarities between our work and the concept of \emph{ownership types},
as well studied for object-oriented languages~\cite{DBLP:conf/oopsla/ClarkePN98}. 
Although in principle ownership types focus on static ownership structures, assessing their 
use for disciplining dynamic authorizations is interesting future work.

\paragraph{\textbf{Acknowledgments.}}
We thank the anonymous referees for their insightful and useful remarks.
This work was supported by COST Action IC1201: Behavioural Types for Reliable Large-Scale Software Systems (BETTY)
via Short-Term Scientific Mission grants (to Pantovi\'c and Vieira).
P\'{e}rez is also affiliated to the NOVA Laboratory for Computer Science and Informatics (NOVA LINCS), 
Universidade Nova de Lisboa, Portugal.

\bibliographystyle{eptcs}

\end{document}